\documentclass{article}
\usepackage{romp}

\usepackage{float}
\usepackage{amsmath}
\usepackage{amsthm}
\usepackage{amssymb}
\usepackage{graphicx}
\usepackage[utf8]{inputenc}

\title{Imaginarity measures induced by relative entropy}
\author{Xiangyu Chen\\School of Mathematics, Harbin Institute of Technology, Harbin 150001, China\\23S012014@stu.hit.edu.cn\\[2ex]
	Qiang Lei\thanks{\quad Supported by National Natural Science Foundation of China (Grants No. 12271474).}
	\\School of Mathematics, Harbin Institute of Technology, Harbin 150001, China
	\\leiqiang@hit.edu.cn}

\begin{document}
		
		\maketitle
		
		\begin{abstract}
		In this paper, we introduce two measures for the resource theory of imaginarity.  One is induced by $\alpha$--$z$--R\'enyi relative entropy and the other, defined for positive definite density matrices, is induced by Tsallis relative operator entropy. The relationships between different imaginarity measures and their properties are also discussed.
		\end{abstract}
		
		
		\section{Introduction}\label{se1}
		
		Quantum resource theories provide an operational way of dealing with the quantification and manipulation of resources in quantum systems, where resources can refer to various quantum properties such as entanglement \cite{Plenio2005}, coherence \cite{Baumgratz2014}, and so on. Quantum resource theories form an important branch in quantum information theory. Recently, a resource theory of imaginarity has been proposed in \cite{Hickey2018}, which has garnered a significant attention.
		\par 
		Let $\mathcal{H}$ be a $d$-dimensional Hilbert space and $\left\lbrace |m\rangle \right\rbrace_{m=0}^{d-1}$ be a fixed set of orthonormal basis on $\mathcal{H}$. $L(\mathcal{H})$ denotes the set of density matrices on $\mathcal{H}$. In the imaginarity resource theory, the set of free states, called real states, is defined as:
		\begin{equation}
			\mathcal{R}=\left\lbrace \rho\in L(\mathcal{H}):\langle  m|\rho|n \rangle\in\mathbb{R}, \forall 
			m,n\in\{0,\ldots,d-1\}\right\rbrace.
			\nonumber
		\end{equation}
		A free operation $\Lambda$, called a real operation, is defined as $\Lambda(\cdot)=\sum_jK_j\cdot K_j^{\dagger}$, where Kraus operators $\left\lbrace K_j \right\rbrace$ sastisfy $\langle  m|K_j|n \rangle\in\mathbb{R}, \forall m,n\in\{0,\ldots,d-1\}$.
		\par 
		In order to quantify imaginarity, an imaginarity measure $\mathcal{M}$ has been introduced by Hickey and Gour in \cite{Hickey2018}. Generally, to quantify imaginarity, $\mathcal{M}$ should satisfy the following conditions:
		\par\noindent
		$(I_1)$: $\mathcal{M}(\rho)\geqslant 0$ for any quantum state $\rho$, and $\mathcal{M}(\rho)=0\iff \rho\in\mathcal{R}$, that is, $\rho$ equals its complex conjugate $\rho^*$.
		\par\noindent
		$(I_2)$: $\mathcal{M}(\Lambda(\rho))\leqslant\mathcal{M}(\rho)$, where $\Lambda$ is a real operation.
		\par\noindent
		$(I_3)$: $\mathcal{M}(\rho)\geqslant\sum_jp_j\mathcal{M}(\rho_j)$, where $p_j=\mathrm{Tr}(K_j\rho K_j^{\dagger})$, $\rho_j=K_j\rho K_j^{\dagger}/p_j$, $\left\lbrace K_j\right\rbrace $ are Kraus operators of a real operation $\Lambda$.
		\par\noindent
		$(I_4)$: $\sum_jp_j\mathcal{M}(\rho_j)\geqslant\mathcal{M}(\sum_jp_j\rho_j)$ for any $\left\lbrace \rho_j \right\rbrace$ and $\left\lbrace p_j \right\rbrace $, where $p_j\geqslant 0$ and $\sum_jp_j=1$.
		\par 
		Similarly to a framework for quantifying coherence \cite{Yu2016}, under the premises of $(I_1)$ and $(I_2)$, the condition $(I_3)+(I_4)$ is equivalent to the following condition \cite{Xue2021}:
		\par\noindent
		$(I_5)$: $\mathcal{M}(p_1\rho_1\oplus p_2\rho_2)=p_1\mathcal{M}(\rho_1)+p_2\mathcal{M}(\rho_2)$ for any $p_1,p_2$ with $p_1+p_2=1$ and $\rho_1,\rho_2\in\mathcal{R}$.
		\par
		Until now, besides the robustness of imaginarity and the $l_1$ measure of imaginarity given in \cite{Hickey2018}, many other imaginarity measures have been given, for example, the fidelity of imaginarity \cite{Wu2020}, the geometric imaginarity measure \cite{Wu2023}, and the convex roof \cite{Chen2022}. In particular, imaginary measures can also be defined by Umegaki relative entropy \cite{Xue2021} and Tsallis entropy \cite{Xu2023} as:
		\begin{equation}
			\mathcal{M}^V(\rho)=S(\rho||\rho^*)=S[\frac{1}{2}(\rho+\rho^*)]-S(\rho)
			\nonumber
		\end{equation}
		and
		\begin{equation}
			\mathcal{M}_q^T(\rho)=(1-q)K_q(\rho||\rho^*),
			\nonumber
		\end{equation}
		where $S(\rho||\sigma)=\mathrm{Tr}[\rho \mathrm{log}(\rho)]-\mathrm{Tr}[\sigma \mathrm{log}(\sigma)]$,  $K_q(\rho||\sigma)=\frac{1}{1-q}[1-\mathrm{Tr}(\rho^{q}\sigma^{1-q})]$, while $\rho$, $\sigma$ are density matrices, and $\rho^*$ represents the conjugation of $\rho$. The value of the parameter $q$ should belong to the interval $(0,1)$.
		\par 
		In addition to these two relative entropies above, other quantifiers have been used to quantify imaginarity, such as R\'enyi entropy \cite{Petz1986,MuellerLennert2013}, Tsallis entropy \cite{Abe2003,Furuichi2005,Nikoufar2017}, Br\`egman relative entropies \cite{Petz2007,Kostecki2017}, and so on. Moreover, recent research [17--19] has shown that some linear combinations or functions of different quantifiers can also be used as new quantifiers.
		\par 
		In this paper, we focus on imaginarity measures induced by relative entropy. In Section 2, we introduce a quantifier induced by $\alpha$--$z$--Rényi relative entropy which is proved to be a imaginarity measure. Furthermore, we provide properties and examples of the measure. In Section 3, we discuss the relationship between different imaginarity measures. In Section 4, we give an imaginarity measure induced by Tsallis relative operator entropy that holds only for positive definite density matrices. Finally, in Section 5, we provide a brief summary.
		
		\section{Imaginarity measure induced by $\alpha$--$z$--R\'enyi relative entropy}\label{se2}
		
		Different from two specific single-parameter relative entropies $S(\rho||\sigma)$ and $K_q(\rho||\sigma)$ mentioned in Section 1, the $\alpha$--$z$--R\'enyi relative entropy, introduced in \cite{Audenaert2015}, is defined by two parameters $\alpha$ and $z$:
		\begin{equation}
			D_{\alpha,z}(\rho||\sigma):=\frac{1}{\alpha-1}\mathrm{log}f_{\alpha,z}(\rho,\sigma),
			\nonumber
		\end{equation}
		where $f_{\alpha,z}(\rho,\sigma)=\mathrm{Tr}(\sigma^{\frac{1-\alpha}{2z}}\rho^{\frac{\alpha}{z}}\sigma^{\frac{1-\alpha}{2z}})^z$. When $z=\alpha$, $D_{\alpha,z}(\rho||\sigma)$ reduces to the quantum R\'enyi divergence $D_{\alpha}(\rho||\sigma)$ \cite{MuellerLennert2013}.
		
		\subsection{Construction of the measure}
		
		In the coherence resource theory, the quantifier inducing coherence measure can be extended from R\'enyi relative entropy to $\alpha$--$z$--R\'enyi relative entropy \cite{Zhu2019}. We will construct a similar extention for quantum imaginarity resource.
		\begin{lemma}{Lemma }\label{LM1}$\mathrm{\cite{Audenaert2007,Bhatia2018}}$
			Suppose $A$, $B$ are two positive semidefinite matrices, $0<t<1$, $r\geqslant 1$, and $q\geqslant 0$, then the following inequality holds:
			\par 
			(1) $\mathrm{Tr}(A^{\frac{1-t}{2t}}BA^{\frac{1-t}{2t}})^t\leqslant \mathrm{Tr}[(1-t)A+tB]$,
			\par 
			(2) $\mathrm{Tr}(A^rB^rA^r)^q\geqslant \mathrm{Tr}(ABA)^{rq}$.
		\end{lemma}
		
		\begin{theorem}{Theorem}\label{TH1}
			The parametrized function $\mathcal{M}_{\alpha,z}^R$ of the state $\rho$, given by: 
			\begin{equation}
				\mathcal{M}_{\alpha,z}^R(\rho)=1-f_{\alpha,z}(\rho,\rho^*),
				\nonumber
			\end{equation}
			where $0<\mathrm{max}\left\lbrace \alpha,1-\alpha\right\rbrace \leqslant z<1$ and $f_{\alpha,z}(\rho,\sigma)=\mathrm{Tr}(\sigma^{\frac{1-\alpha}{2z}}\rho^{\frac{\alpha}{z}}\sigma^{\frac{1-\alpha}{2z}})^z$, is an imaginarity measure.
		\end{theorem}
		\begin{proof}
			We only need to prove $(I_1) + (I_2) + (I_5)$.
			\par\noindent
			$(I_1)$: The trace of the density matrix $\rho$ is 1, and the same applies to $\rho^*$, so from Lemma 1 it follows:
			\begin{align*}
				f_{\alpha,z}(\rho,\rho^*)
				&=\mathrm{Tr}[(\rho^*)^{\frac{1-\alpha}{2z}}\rho^{\frac{\alpha}{z}}(\rho^*)^{\frac{1-\alpha}{2z}}]^z
				\\
				&=
				 \mathrm{Tr}[(\rho^*)^{\frac{1-\alpha}{2z}}\rho^{\frac{\alpha}{z}}(\rho^*)^{\frac{1-\alpha}{2z}}]^{\frac{z}{\alpha}\alpha}
				\\
				&\leqslant \mathrm{Tr}[(\rho^*)^{\frac{1-\alpha}{2\alpha}}\rho^{\frac{\alpha}{\alpha}}(\rho^*)^{\frac{1-\alpha}{2\alpha}}]^z
				\\
				&\leqslant \mathrm{Tr}[(1-\alpha)(\rho^*)+\alpha \rho]^z
				\\
				&=1.
				\nonumber
			\end{align*}
			This leads to $\mathcal{M}_{\alpha,z}^R(\rho)\geqslant 0$. When the equality holds, an equivalence chain exists:
			\begin{align*}
				\mathcal{M}_{\alpha,z}^R(\rho)=0\iff f_{\alpha,z}(\rho,\rho^*)=1
				\\
				\iff D_{\alpha,z}(\rho||\rho^*)=0\iff \rho=\rho^*.
				\nonumber
			\end{align*}
			This implies $\rho\in\mathcal{R}$. Thus the condition $(I_1)$ is satisfied.
			\par\noindent
			$(I_2)$: Since the density matrix $\rho$ can be written as $\mathrm{Re}(\rho)+i\mathrm{Im}(\rho)$, we can know that $\Lambda(\rho^*)=[\Lambda(\rho)]^*$ by using Kraus operators. For any completely positive trace-preserving (CPTP) map $\Lambda$, and for any density matrices 
			$\rho$ and $\sigma$, the following inequality holds \cite{Hiai2013}:
			\begin{equation}
				D_{\alpha,z}(\Lambda(\rho)||\Lambda(\sigma))\leqslant D_{\alpha,z}(\rho,\sigma).
				\nonumber
			\end{equation}
			The condition $(I_2)$ holds since
			\begin{equation}
				\mathcal{M}_{\alpha,z}^R(\rho)=1-e^{(\alpha-1)D_{\alpha,z}(\rho,\rho^*)}.
				\nonumber
			\end{equation}
			\par\noindent
			$(I_5)$: Set $\rho=p_1\rho_1\oplus p_2\rho_2$. Then:
			\begin{eqnarray}
				&&f_{\alpha,z}(p_1\rho_1\oplus p_2\rho_2,p_1\rho_1^*\oplus p_2\rho_2^*)
				\nonumber
				\\
				&=&\mathrm{Tr}\left[ \bigoplus_{j=1}^2 p_j^{\frac{1-\alpha}{2z}}p_j^{\frac{\alpha}{z}}p_j^{\frac{1-\alpha}{2z}}(\rho_j^*)^{\frac{1-\alpha}{2z}}\rho_j^{\frac{\alpha}{z}}(\rho_j^*)^{\frac{1-\alpha}{2z}}\right]^z
				\nonumber
				\\
				&=&\sum_{j=1}^{2}p_j\mathrm{Tr}\left[ (\rho_j^*)^{\frac{1-\alpha}{2z}}\rho_j^{\frac{\alpha}{z}}(\rho_j^*)^{\frac{1-\alpha}{2z}}\right]^z
				\nonumber
				\\
				&=&p_1f_{\alpha,z}(\rho_1,\rho_1^*)+p_2f_{\alpha,z}(\rho_2,\rho_2^*).
				\nonumber
			\end{eqnarray}
			For $p_1+p_2=1$, we obtain:
			\begin{equation}
				\mathcal{M}_{\alpha,z}^R(p_1\rho_1\oplus p_2\rho_2)=p_1\mathcal{M}_{\alpha,z}^R(\rho_1)+p_2\mathcal{M}_{\alpha,z}^R(\rho_2).
				\nonumber
			\end{equation}
			Thus $\mathcal{M}_{\alpha,z}^R$ satisfies $(I_5)$.
		\end{proof}
		\par 
		 The range of parameters $\alpha$ and $z$, $0<\mathrm{max}\left\lbrace \alpha,1-\alpha\right\rbrace \leqslant z<1$, is a standing assumption for the rest of this paper.
		 
		\subsection{Properties of $\mathcal{M}_{\alpha,z}^R$}
	
		Relying on the axiomatic properties of $\alpha$--$z$--R\'{e}nyi relative entropy \cite{Audenaert2015}, we can obtain the following results.
		\begin{theorem}{Theorem}\label{TH2}
			The imaginarity measure $\mathcal{M}_{\alpha,z}^R$ has the following properties:
			\par
			(1) $\mathcal{M}_{\alpha,z}^R$ is invariant under any unitary matrix $U$.
			\par
			(2) For any density matrix $\tau$, we have $\mathcal{M}_{\alpha,z}^R(\rho\otimes\tau)\geqslant\mathcal{M}_{\alpha,z}^R(\rho)\mathcal{M}_{\alpha,z}^R(\tau)$, which means that imaginarity measure increases under the tensor product.
		\end{theorem}
		\begin{proof}
			Here, for any density matrices $\rho$ and $\tau$, we have:
			\begin{eqnarray}
				&&f_{\alpha,z}(U\rho U^\dagger,U\rho^*U^\dagger)
				\nonumber
				\\
				&=&\mathrm{Tr}\left[ (U\rho^* U^\dagger)^{\frac{1-\alpha}{2z}}(U\rho U^\dagger)^{\frac{\alpha}{z}}(U\rho^* U^\dagger)^{\frac{1-\alpha}{2z}}\right] ^z
				\nonumber
				\\
				&=&\mathrm{Tr}\left[ U(\rho^*)^{\frac{1-\alpha}{2z}}\rho^{\frac{\alpha}{z}}(\rho^*)^{\frac{1-\alpha}{2z}}U^\dagger\right]^z
				\nonumber
				\\
				&=&\mathrm{Tr} \left\lbrace U\left[ (\rho^*)^{\frac{1-\alpha}{2z}}\rho^{\frac{\alpha}{z}}(\rho^*)^{\frac{1-\alpha}{2z}}\right] ^zU^\dagger\right\rbrace 
				\nonumber
				\\
				&=&f_{\alpha,z}(\rho,\rho^*).
				\nonumber
			\end{eqnarray}
			Therefore, $\mathcal{M}_{\alpha,z}^R(U\rho U^*)=\mathcal{M}_{\alpha,z}^R(\rho)$ holds.
			\begin{eqnarray}
				&&f_{\alpha,z}(\rho\otimes\tau,\rho^*\otimes\tau^*)
				\nonumber
				\\
				&=&\mathrm{Tr}\left[ (\rho^*\otimes\tau^*)^{\frac{1-\alpha}{2z}}(\rho\otimes\tau)^{\frac{\alpha}{z}}(\rho^*\otimes\tau^*)^{\frac{1-\alpha}{2z}}\right] ^z
				\nonumber
				\\
				&=&\mathrm{Tr}\left\lbrace \left[ (\rho^*)^{\frac{1-\alpha}{2z}}\otimes(\tau^*)^{\frac{1-\alpha}{2z}}\right] (\rho^{\frac{\alpha}{z}}\otimes\tau^{\frac{\alpha}{z}})\left[ (\rho^*)^{\frac{1-\alpha}{2z}}\otimes(\tau^*)^{\frac{1-\alpha}{2z}}\right] \right\rbrace ^z
				\nonumber
				\\
				&=&\mathrm{Tr}\left\lbrace \left[ 
				(\rho^*)^{\frac{1-\alpha}{2z}}\rho^{\frac{\alpha}{z}}(\rho^*)^{\frac{1-\alpha}{2z}}\right]\otimes\left[(\tau^*)^{\frac{1-\alpha}{2z}}\tau^{\frac{\alpha}{z}}(\tau^*)^{\frac{1-\alpha}{2z}}\right]\right\rbrace ^z
				\nonumber
				\\
				&=&\mathrm{Tr}\left\lbrace \left[ 
				(\rho^*)^{\frac{1-\alpha}{2z}}\rho^{\frac{\alpha}{z}}(\rho^*)^{\frac{1-\alpha}{2z}}\right]^z\otimes\left[(\tau^*)^{\frac{1-\alpha}{2z}}\tau^{\frac{\alpha}{z}}(\tau^*)^{\frac{1-\alpha}{2z}}\right]^z\right\rbrace
				\nonumber
				\\
				&=&\mathrm{Tr}\left[ 
				(\rho^*)^{\frac{1-\alpha}{2z}}\rho^{\frac{\alpha}{z}}(\rho^*)^{\frac{1-\alpha}{2z}}\right]^z\mathrm{Tr}\left[(\tau^*)^{\frac{1-\alpha}{2z}}\tau^{\frac{\alpha}{z}}(\tau^*)^{\frac{1-\alpha}{2z}}\right]^z
				\nonumber
				\\
				&=&f_{\alpha,z}(\rho,\rho^*)f_{\alpha,z}(\tau,\tau^*).
				\nonumber
			\end{eqnarray}
			Due to $f_{\alpha,z}(\rho,\rho^*)\leqslant 1$, $f_{\alpha,z}(\tau,\tau^*) \leqslant 1$, we can know that $2f_{\alpha,z}(\rho,\rho^*)f_{\alpha,z}(\tau,\tau^*)\leqslant f_{\alpha,z}(\rho,\rho^*)+f_{\alpha,z}(\tau,\tau^*)$. Then $1-f_{\alpha,z}(\rho\otimes\tau,\rho^*\otimes\tau^*)\geqslant [1-f_{\alpha,z}(\rho,\rho^*)][1-f_{\alpha,z}(\tau,\tau^*)]$ holds, implying property (2).
		\end{proof}

		\section{Relationships between imaginarity measures}\label{se3}
		
		\subsection{Monotonicity with respect to parameters}
		
		For simplicity, we denote $\mathcal{M}_{\alpha,z}^R$ as $\mathcal{M}_{\alpha}^R$ and $D_{\alpha,z}$ as $D_{\alpha}$ when $\alpha=z$. 
		\begin{theorem}{Theorem}\label{TH3}
			The imaginarity measure $\mathcal{M}_{\alpha,z}^R$ satisfies:
			\par
			(1) $\mathcal{M}_{\alpha,z}^R=\mathcal{M}_{1-\alpha,z}^R$;
			\par
			(2) $\mathcal{M}_{\alpha_1}^R\leqslant\mathcal{M}_{\alpha_2}^R$ if $\alpha_1\leqslant\alpha_2$;
			\par
			(3) $\mathcal{M}_{\alpha,z_1}^R\leqslant\mathcal{M}_{\alpha,z_2}^R$ if $z_1\leqslant z_2$.
		\end{theorem}
		\begin{proof}
			We can obtain (1) by taking the conjugate:
			\begin{align}
				\left\lbrace \mathrm{Tr}\left[ (\rho^*)^{\frac{1-\alpha}{2z}}\rho^{\frac{\alpha}{z}}(\rho^*)^{\frac{1-\alpha}{2z}}\right] ^z\right\rbrace^*=\mathrm{Tr}\left[\rho^{\frac{1-\alpha}{2z}}(\rho^*)^{\frac{\alpha}{z}}\rho^{\frac{1-\alpha}{2z}}\right]^z=\mathrm{Tr}\left[(\rho^*)^{\frac{\alpha}{2z}}\rho^{\frac{1-\alpha}{z}}(\rho^*)^{\frac{\alpha}{2z}}\right]^z.
				\nonumber
			\end{align}
			If $\alpha_1\leqslant \alpha_2$, then $D_{\alpha_1}(\rho||\sigma)\leqslant D_{\alpha_2}(\rho||\sigma)$ \cite{MuellerLennert2013}. From $\alpha-1<0$ and monotonicity of the logarithmic function, we have $f_{\alpha_1,\alpha_1}\geqslant f_{\alpha_2,\alpha_2}$, which gives us (2). The proof of (3) requires only to consider $\frac{z_2}{z_1}\geqslant 1$, which directly leads to $f_{\alpha,z_1}\geqslant f_{\alpha,z_2}$ by Lemma 1.
		\end{proof}
		\begin{theorem}{Theorem}\label{TH4}
			The imaginarity measure $\mathcal{M}_q^T$ satisfies:
			\par
			(1) $\mathcal{M}_q^T(\rho)=\mathcal{M}_{1-q}^T(\rho)$;
			\par
			(2) If $q_1\leqslant q_2\leqslant\frac{1}{2}$, then $\mathcal{M}_{q_1}^T(\rho)\leqslant\mathcal{M}_{q_2}^T(\rho)$.
		\end{theorem}
		\begin{proof}
			The proof of property (1) is similar to the proof of Theorem 3.(1). It remains to prove (2). From \cite{Bhattacharya2023} we know that $\mathrm{Tr}(\rho^q(\rho^*)^{1-q})$ is a log-convex function when $q\in (0,1)$. A log-convex function must be convex \cite{Conway1978}. Combining with (1), we can conclude the inequality $\mathrm{Tr}(\rho^{q_1}(\rho^*)^{1-q_1})\geqslant \mathrm{Tr}(\rho^{q_2}(\rho^*)^{1-q_2})$ holds when $q_1\leqslant q_2\leqslant\frac{1}{2}$.
		\end{proof}
		
		\subsection{The magnitude relation between the two imaginarity measures}
		
		First, let's compare the relationship between $\mathcal{M}_{\alpha,z}^R$ and $\mathcal{M}_q^T$. Since $\frac{1}{z}>1$, we can infer from Lemma 1 that:
		\begin{align*}
			\mathrm{Tr}(\sigma^{\frac{1-\alpha}{2z}}\rho^{\frac{\alpha}{z}}\sigma^{\frac{1-\alpha}{2z}})^z&\geqslant \mathrm{Tr}(\sigma^{\frac{1-\alpha}{2}}\rho^{\alpha}\sigma^{\frac{1-\alpha}{2}})^{\frac{1}{z}z}
			\\
			&=\mathrm{Tr}(\sigma^{\frac{1-\alpha}{2}}\rho^{\alpha}\sigma^{\frac{1-\alpha}{2}})
			\\
			&=\mathrm{Tr}(\rho^{\alpha}\sigma^{1-\alpha}).
		\end{align*}
		Therefore, in combination with Theorem 3 and Theorem 4, we can obtain the following theorem:
		\begin{theorem}{Theorem}\label{TH5}
			For any quantum state $\rho$,
			\begin{equation}
				\mathcal{M}_{\alpha}^R(\rho)\leqslant\mathcal{M}_{\alpha,z}^R(\rho)\leqslant\mathcal{M}_{\alpha}^T(\rho).
				\nonumber
			\end{equation}
		\end{theorem}
		\par 
		Using Theorem 5, we can derive some inequalities, such as:
		\par
		(1) $\mathcal{M}_{\alpha}^R(\rho\otimes\tau)\leqslant \mathcal{M}_{\alpha}^T(\rho\otimes\tau)$;
		\par
		(2) $\mathcal{M}_{\alpha}^R(\rho^{\otimes n})\leqslant \mathcal{M}_{\alpha}^T(\rho^{\otimes n})$.
		\par 
		However, when $q\neq \alpha$, $\mathcal{M}_{q}^T$ can not be simply included in Theorem 5.
		\begin{example}{Example}
			Take the following density matrix as an exmple:
			\begin{equation}
				\rho_0=\frac{1}{10}
				\begin{pmatrix}
					4 & 3-i
					\\
					3+i & 6
				\end{pmatrix}.
				\nonumber
			\end{equation}
			Then
			\begin{equation}
				\mathcal{M}_{0.3}^T(\rho_0)\leqslant\mathcal{M}_{0.5}^R(\rho_0)\leqslant\mathcal{M}_{0.5}^T(\rho_0).
				\nonumber
			\end{equation}
		\end{example}
		
		For the composite systems, we have:
		\begin{theorem}{Theorem}\label{TH6}
			For any density matrix $\rho$ and $\tau$, for any $q$, $\alpha$, and $z$ satisfying $q\in(0,1)$ and $0<\mathrm{max}\left\lbrace \alpha,1-\alpha\right\rbrace \leqslant z<1$, if $\mathcal{M}_{\alpha,z}^R(\rho)\leqslant \mathcal{M}_q^T(\rho)$ and $\mathcal{M}_{\alpha,z}^R(\tau)\leqslant \mathcal{M}_q^T(\tau)$ holds, then $\mathcal{M}_{\alpha,z}^R(\rho\otimes\tau)\leqslant \mathcal{M}_q^T(\rho\otimes\tau)$ holds.
		\end{theorem}
		\begin{proof}
			Follows directly from the 
			definitions of $\mathcal{M}^R_{\alpha,z}$ and $\mathcal{M}^T_{\alpha}$. In fact, from the conditions we can know that $f_{\alpha,z}(\rho,\rho^*)\geqslant\mathrm{Tr}[\rho^{q}(\rho^*)^{1-q}]$ and $f_{\alpha,z}(\tau,\tau^*)\geqslant\mathrm{Tr}[\tau^{q}(\tau^*)^{1-q}]$. Then $f_{\alpha,z}(\rho\otimes\tau,\rho^*\otimes\tau^*)\geqslant\mathrm{Tr}[(\rho\otimes\tau)^{q}(\rho^*\otimes\tau^*)^{1-q}]$. This leads to the result.
		\end{proof}
		
		\section{Imaginarity measure of positive definite density matrices}\label{se4}
		
		As mentioned in Section 1, some parameterized functions cannot be considered as appropriate resource measures (cf. also \cite{Zhao2018}). In order to tackle this problem, in the following, we propose a new approach for measures of quantum imaginarity. However, the domain will be changed from all quantum states to positive definite density matrices.
		\par 
		In order to distinguish positive definite density matrices, we will denote them by $\delta$ instead of $\rho$.
		
		\subsection{Imaginarity measure induced by Tsallis relative operator entropy}
		
		In \cite{Furuichi2005} the authors introduced the Tsallis relative operator entropy, defined as follows:
		\begin{equation}
			T_{\lambda}(\delta||\eta):=\delta^{\frac{1}{2}}\mathrm{ln}_{\lambda}(\delta^{-\frac{1}{2}}\eta\delta^{-\frac{1}{2}})\delta^{\frac{1}{2}},
			\nonumber
		\end{equation}
		where $\delta$ and $\eta$ are two invertible positive operators, $\mathrm{ln}_{\lambda}X=\frac{X^{\lambda}-I}{\lambda}$ for positive definite operator $X$ and identity operator $I$ with $\lambda\in (0,1]$. We can rewrite $T_{\lambda}(\delta||\eta)$ in another form by introducing the notation $\delta\sharp_{\lambda}\eta=\delta^{\frac{1}{2}}(\delta^{-\frac{1}{2}}\eta\delta^{-\frac{1}{2}})^\lambda\delta^{\frac{1}{2}}$:
		\begin{equation}
			T_{\lambda}(\delta||\eta)=\frac{1}{\lambda}(\delta\sharp_{\lambda}\eta-\delta).
			\nonumber
		\end{equation}
		\par
		We observe that similar to the relationship between $\mathcal{M}_{\alpha,z}^R(\delta)$ and $\mathcal{M}^V(\delta)$, the extension of $\mathcal{M}_q^T(\delta)$ is also possible.
		\begin{lemma}{Lemma}\label{LM2}$\mathrm{\cite{Furuichi2005}}$
			For positive definite density matrices $\delta$ and $\eta$, a real parameter $a$ with $a>0$, and $\lambda\in(0,1]$, the following inequalities hold:
			\begin{equation}
				\begin{cases}
					T_{\lambda}(\delta||\eta)\geqslant\delta\sharp_{\lambda}\eta-\frac{1}{a}\delta\sharp_{\lambda-1}\eta+(\mathrm{ln}_\lambda\frac{1}{a})\delta,
					\\
					T_{\lambda}(\delta||\eta)\leqslant\frac{1}{a}\eta-\delta-(\mathrm{ln}_{\lambda}\frac{1}{a})\delta\sharp_{\lambda}\eta,
				\end{cases}
				\nonumber
			\end{equation}
			where $\mathrm{ln}_{\lambda}\frac{1}{a}=\frac{(\frac{1}{a})^{\lambda}-1}{\lambda}$. Morever, $T_{\lambda}(\delta||\eta)=0$ if and only if $\delta=\eta$ \cite{Fujii1989}.
		\end{lemma}
		\begin{lemma}{Lemma}\label{LM3}$\mathrm{\cite{Ando1979}}$
			If $\Phi$ is a positive linear map, then, for any positive definite matrices $A$ and $B$:
			\begin{equation}
				\Phi(A\sharp_{\lambda}B)\leqslant \Phi(A)\sharp_{\lambda}\Phi(B).
				\nonumber
			\end{equation}
		\end{lemma}
		\begin{theorem}{Theorem}\label{TH7}
			The function
			\begin{equation}
				\mathcal{M}_{\lambda}^O(\delta)=1-\mathrm{Tr}(\delta\sharp_{\lambda}\delta^*),
				\nonumber
			\end{equation}
			where $\lambda\in (0,1)$, is an imaginarity measure: 
		\end{theorem}
		\begin{proof}
			
			First, we can see $\mathrm{Tr}(\delta\sharp_{\alpha}\delta^*)\leqslant\mathrm{Tr}[(1-\alpha)\delta+\alpha \delta^*]=1$ in \cite{Bhatia2018}, thus $\mathcal{M}_{\alpha}^O(\delta)\geqslant 0$. In conjunction with the discussion from Lemma 2, the condition $(I_1)$ is satisfied.
			\par 
			Using Lemma 3, $\Phi(\delta^*)=[\Phi(\delta)]^*$ and the fact that $\Phi$ is trace-preserving, we obtain:
			\begin{eqnarray*}
				\mathrm{Tr}(\delta\sharp_{\alpha}\delta^*)=\mathrm{Tr}[\Phi(\delta\sharp_{\alpha}\delta^*)]\leqslant\mathrm{Tr}[ \Phi(\delta)\sharp_{\alpha}\Phi(\delta^*)]=\mathrm{Tr}[\Phi(\delta)\sharp_{\alpha}\Phi(\delta)^*].
			\end{eqnarray*}
			Thus, the condition $(I_2)$ is satisfied.
			\par
			Suppose $\delta=d_1\delta_1\oplus d_2\delta_2$ with $d_1+d_2=1$. Then:
			\begin{eqnarray*}
				&&\mathcal{M}_{\alpha}^O(\delta)=1-\mathrm{Tr}(\delta\sharp_{\alpha}\delta^*)=1-\mathrm{Tr}[(d_1\delta_1\oplus d_2\delta_2)\sharp_{\alpha}(d_1\delta_1^*\oplus d_2\delta_2^*)]
				\\
				&=&1-\mathrm{Tr}\left\lbrace (d_1\delta_1\oplus d_2\delta_2)^{\frac{1}{2}}\left[ (d_1\delta_1\oplus d_2\delta_2)^{-\frac{1}{2}}(d_1\delta_1^*\oplus d_2\delta_2^*)(d_1\delta_1\oplus d_2\delta_2)^{-\frac{1}{2}}\right]^\alpha (d_1\delta_1\oplus d_2\delta_2)^{\frac{1}{2}}\right\rbrace
				\\
				&=&1-\mathrm{Tr}\left[ \bigoplus_{j=1}^2 d_j^{\frac{1}{2}}(d_j^{-\frac{1}{2}}d_jd_j^{-\frac{1}{2}})^{\alpha}d_j^{\frac{1}{2}}\delta_j^{\frac{1}{2}}(\delta_j^{-\frac{1}{2}}\delta_j^*\delta_j^{-\frac{1}{2}})^{\alpha}\delta_j^{\frac{1}{2}}\right]
				\\
				&=&1-\mathrm{Tr}\left[ \bigoplus_{j=1}^2 d_j\delta_j^{\frac{1}{2}}(\delta_j^{-\frac{1}{2}}\delta_j^*\delta_j^{-\frac{1}{2}})^{\alpha}\delta_j^{\frac{1}{2}}\right]
				\\
				&=&d_1\left\lbrace 1-\mathrm{Tr}[\delta_1^{\frac{1}{2}}(\delta_1^{-\frac{1}{2}}\delta_1^*\delta_1^{-\frac{1}{2}})^{\alpha}\delta_1^{\frac{1}{2}}]\right\rbrace+d_2\left\lbrace 1-\mathrm{Tr}[\delta_2^{\frac{1}{2}}(\delta_2^{-\frac{1}{2}}\delta_2^*\delta_2^{-\frac{1}{2}})^{\alpha}\delta_2^{\frac{1}{2}}]\right\rbrace
				\\
				&=&d_1\mathrm{Tr}(\delta_1\sharp\delta_1^*)+d_2\mathrm{Tr}(\delta_2\sharp\delta_2^*)=d_1\mathcal{M}_{\alpha}^O(\delta_1)+d_2\mathcal{M}_{\alpha}^O(\delta_2).
			\end{eqnarray*}
			\par
			This proves the condition $(I_5)$.
		\end{proof}
		
		\subsection{$\mathcal{M}_{\lambda}^O$ versus $\mathcal{M}_{\alpha,z}^R$ and $\mathcal{M}_{\lambda}^T$}
		
		In our framework, these measures can be compared only for positive definite 
		density matrices.
		\par 
		From \cite{Bhatia2018}, we have the inequality:
		\begin{equation}
			\mathrm{Tr}(\delta\sharp_{\lambda}\delta^*)\leqslant\mathrm{Tr}[\delta^{1-\lambda}(\delta^*)^{\lambda}]=\mathrm{Tr}[\delta^{\lambda}(\delta^*)^{1-\lambda}].
			\nonumber
		\end{equation}
		Therefore, we can establish a connection between $\mathcal{M}_{\lambda}^O$, $\mathcal{M}_{\alpha,z}^R$, and $\mathcal{M}_{\lambda}^T$.
		\begin{theorem}{Theorem}\label{TH8}
			For positive definite density matrices $\delta$,
			\begin{equation}
				\mathcal{M}_{\alpha}^R(\delta)\leqslant\mathcal{M}_{\alpha,z}^R(\delta)\leqslant\mathcal{M}_{\alpha}^T(\delta)\leqslant\mathcal{M}_{\alpha}^O(\delta).
				\nonumber
			\end{equation}
		\end{theorem}
		However, when $\lambda\neq\alpha$, the situation becomes more complicated.
		\begin{example}{Example}
			Consider a positive definite density matrix
			\begin{equation}
				\delta_0=\frac{1}{10}
				\begin{pmatrix}
					6 & 1+i
					\\
					1+i & 4
				\end{pmatrix}.
				\nonumber
			\end{equation}
			Then
			\begin{equation}
				\mathcal{M}_{0.3}^O(\delta_0)\leqslant\mathcal{M}_{0.5}^R(\delta_0)\leqslant\mathcal{M}_{0.5}^T(\delta_0).
				\nonumber
			\end{equation}
		\end{example}
		
		\section{Conclusion}\label{se5}
		
		We have introduced a new imaginarity measure, $\mathcal{M}^R_{\alpha,z}$. Its relevant properties and a comparison of $\mathcal{M}^R_{\alpha,z}$ and $\mathcal{M}^T_q$ were provided. We have also introduced an imaginarity measure $\mathcal{M}^O_{\lambda}$ for positive definite density matrices, and compared it with the above measures.

\end{document}